\documentclass[11pt]{article}
\usepackage[margin=1in]{geometry}

\usepackage{amsmath}
\usepackage{amsthm}
\usepackage{amssymb}
\usepackage{algorithm}
\usepackage{color}
\usepackage[english]{babel}
\usepackage{graphicx}
\usepackage{grffile}
\usepackage{wrapfig,epsfig}
\usepackage{epstopdf}
\usepackage{url}
\usepackage{color}
\usepackage{epstopdf}
\usepackage{algpseudocode}
\usepackage[T1]{fontenc}
\usepackage{bbm}
\usepackage{comment}
\usepackage{dsfont}
\usepackage{thm-restate}
\usepackage{bm}
\usepackage{tcolorbox}
\usepackage{subcaption}
\usepackage{algpseudocode}

\usepackage{enumitem}
\usepackage{hyperref}

\graphicspath{{./figs/}}

\newtheorem{theorem}{Theorem}[section]
 
\newtheorem{lemma}[theorem]{Lemma}

\newtheorem{Remark}[theorem]{Remark}

\newcommand{\eps}{\epsilon}

\newcommand{\R}{\mathbb{R}}

\renewcommand{\varepsilon}{\epsilon}

\renewcommand{\eps}{\epsilon}

\DeclareMathOperator*{\polylog}{{polylog}}

\DeclareMathOperator{\poly}{poly}

\algnewcommand\algorithmicforeach{\textbf{for each}}

\begin{document}

\title{Fully-dynamic-to-incremental reductions with known deletion order (e.g. sliding window)}
\author{Binghui Peng\footnote{Supported by NSF CCF-1703925, IIS-1838154, CCF-2106429 and CCF-2107187.
} \\ Columbia University \\  \texttt{bp2601@columbia.edu} 
\and Aviad Rubinstein\footnote{Supported by NSF CCF-1954927, and a David and Lucile Packard Fellowship.
} \\ Stanford University \\ \texttt{aviad@cs.stanford.edu}
}
\date{}

\maketitle






\begin{abstract}
\small\baselineskip=9pt
Dynamic algorithms come in three main flavors: {\em incremental} (insertions-only), {\em decremental} (deletions-only), or  {\em fully dynamic} (both insertions and deletions). Fully dynamic is the holy grail of dynamic algorithm design; it is obviously more general than the other two, but is it strictly harder?

Several works managed to reduce fully dynamic to the incremental or decremental models by taking advantage of either specific structure of the incremental/decremental algorithms (e.g.~\cite{HK99, HLT01, BKS12, ADKKP16,BS80,OL81,OvL81}), or specific order of insertions/deletions (e.g.~\cite{AW14,HKNS15,KPP16}). Our goal in this work is to get a black-box fully-to-incremental reduction that is as general as possible. We find that the following conditions are necessary:
\begin{itemize}
    \item The incremental algorithm must have a worst-case (rather than amortized) running time guarantee.
    \item The reduction must work in what we call the {\em deletions-look-ahead model}, where the order of deletions among current elements is known in advance. A notable practical example is the ``sliding window'' (FIFO) order of updates. 
\end{itemize}
Under those conditions, we design:
\begin{itemize}
    \item A simple, practical, amortized-fully-dynamic to worst-case-incremental reduction with a $\log(T)$-factor overhead on the running time, where $T$ is the total number of updates. 
    \item A theoretical worst-case-fully-dynamic to worst-case-incremental reduction with a $\polylog(T)$-factor overhead on the running time. 
\end{itemize}
\end{abstract}

\newpage

\section{Introduction}
\label{sec:intro}

A dynamic algorithm is a data structure that maintains certain properties (e.g. shortest paths) of a ground set that is subject to a sequence of updates (e.g. insertions/deletions of edges), and the goal is to minimize the (total) update time.
A {\em fully dynamic} algorithm supports both insertions and deletions of the ground set elements, while an {\em incremental} algorithm restricts the updates to be insertion-only and a {\em decremental} algorithm handles deletion operations only.

A fully dynamic algorithm clearly benefits from handling more general updates, but the at the same time, it is expected to be much harder than incremental/decremental algorithms. 
Meanwhile, several existing works have exploited special structure (\cite{HK99, HLT01, BKS12, ADKKP16, BS80,OL81,OvL81}) or specific order of update sequence (\cite{AW14,HKNS15,KPP16}) and reduce fully dynamic to incremental or decremental algorithm.
This motivates one to ask

\begin{center}
\textit{Can a generic reduction transform an incremental algorithm into one that handles both insertions and deletions? }
\end{center}

 Perhaps surprisingly, we find that once the order of deletions of current elements is known to the algorithm (deletions-look-ahead), one can translate an incremental algorithm with {\em worst case guarantee} to a fully dynamic algorithm with worst case guarantee, with only polylogarithmic overhead.

\begin{theorem}[Reduction, worst case to worst case]
\label{thm:main}
Let $T \geq 1$ be the total number of updates. Suppose there exists a dynamic algorithm in the incremental setting with query time $\Gamma_{q}$ and worst case update time $\Gamma_{u}$, then there is a dynamic algorithm for deletions-look-ahead setting with query time $\Gamma_{q}$ and worst case update time%
\footnote{We believe that the update time can be improved at least to $O(\Gamma_{u} \log(T) \log\log(T ))$ at the expense of a more complicated algorithm.} $O(\Gamma_{u} \log^2 (T))$.
\end{theorem}

Our reduction requires the incremental algorithm to have worst case (rather than amortized) runtime guarantee, and most importantly, the relative order of deletions of current elements must be known.\footnote{Equivalently, one can assume there exists an oracle that outputs the deletion order of existing ground set elements.} The latter assumption is satisfied by well-studied models including sliding window (FIFO) model \cite{epasto2017submodular, epasto2022improved, datar2002maintaining,braverman2020near, woodruff2022tight}, where only most recent insertion updates are of interests, as well as offline/look-ahead model \cite{khanna1998certificates, sankowski2010fast, van2019dynamic, chen2020fast, AW14}, where the entire sequence of updates is known in advance. 
These models are of interests to both theoretical and empirical community \cite{hanauerrecent}. Furthermore, the sliding window model is often used as a benchmark for empirical investigation of fully dynamic algorithms (e.g.~\cite{wang2018location, lattanzi2020fully, henzinger2021random}).
To complement our results, we also prove that both conditions are indeed indispensable for a black box reduction (see Section \ref{sec:impossible}).

If one only aims for algorithms with amortized runtime guarantee, we have a reduction with improved update time. We believe it could be of independent interests due to its simplicity and could be beneficial for empirical implementation.
\begin{theorem}[Reduction, amortized to worst case]
\label{thm:main-amortize}
Let $T \geq 1$ be the total number of updates. Suppose there exists a dynamic algorithm for {\em incremental setting} with query time $\Gamma_{q}$ and worst case update time $\Gamma_{u}$, then there is a dynamic algorithm for {\em deletions-look-ahead} setting with query time $\Gamma_{q}$ and amortized update time $O(\Gamma_{u} \cdot \log (T))$.
\end{theorem}

On the technical side, Theorem~\ref{thm:main-amortize} exploits the ideas of
rewinding the incremental algorithm (e.g.~\cite{AW14,HKNS15,KPP16}), and in particular achieves only logarithmic overhead by rewinding $2^i$ insertions every $\Theta(2^i)$ updates (e.g.~\cite{HK99, HLT01, BKS12, ADKKP16}). 
To achieve the worst case guarantee, Theorem \ref{thm:main} additionally requires amortizing this rewinding of $2^i$ insertions over $\Theta(2^i)$ updates; this is a little trickier because we have to start re-inserting elements in advance before all the elements we would want to insert have arrived. 

We demonstrate the power of our reduction in Section \ref{sec:application} by providing applications on dynamic submodular maximization, dynamic Depth first search (DFS) tree and dynamic All-Pair Shortest-Paths (APSP).

\paragraph{Related work} 
A systematic study of black box reduction for dynamic algorithms has been initiated by the seminal work of \cite{BS80}, which shows how to make a static data structure support insertions for ``decomposable search problem''.
The ideas of rewinding incremental algorithm and the logarithmic scheduling have been studied in different problem specific context \cite{AW14,HKNS15,KPP16,HK99, HLT01, BKS12, ADKKP16,BS80,OL81,OvL81,DS91,Chan12}.
Our deletion-look-ahead model has also been considered in the computational geometry literature, and it is one variant of the semi-online model (see \cite{DS91} for a detailed discussion). 
The work of \cite{Chan12} is closely related to us, it presents a worst-case to amortized-case reduction, when the exact deletion time is known. The result is (almost) equivalent to Theorem \ref{thm:main-amortize}.\footnote{We thank Timothy Chan for pointing out this connection after we published the first version.}
It differs from Theorem \ref{thm:main} as our reduction is worst-case to worst-case.

\paragraph{Notation} Let $[n] = \{1, 2, \ldots, n\}$ and $[n_1:n_2] = \{n_1, n_1+1, \ldots, n_2\}$. 
For any ground set elements $e_1$ and $e_2$, we say $e_1$ is {\em younger} than $e_2$ if $e_1$ would be deleted earlier than $e_2$.
An element is of rank $r$ if it is the $r$-th youngest of the current elements.
For an ordered set $A$, we use $A[n_1:n_2]$ to denote the youngest $n_1$-th to $n_2$-th elements of $A$. 
For any $t \in [T]$, let $k(t)$ be the largest integer such that $t$ is exactly a multiple of $2^{k(t)}$.
In our pseudocode, {\sc Insert} refers to the insertion procedure of the incremental algorithm (hence taking $O(\Gamma_u)$ time), while adding/removing elements to a set $A$ operates only over the the set (hence taking $O(1)$ time).

\section{Reduction -- Amortized to worst case}
\label{sec:reduction-amortize}

We start by providing the amortized to worst case reduction and prove Theorem \ref{thm:main-amortize}. It also serves as a warm up for the worst case to worst case reduction.
Our key idea is to re-order current elements in reverse order of deletions, this ensures a deletion update takes only $O(\Gamma_{u})$ time by rewinding the computation.
Of course, maintaining the reverse order could be expensive as the new-coming element could be deleted last (e.g. the sliding window model). Instead, we only maintain a partial reverse order -- we re-insert the last $O(2^i)$ elements in reverse order in every $2^{i}$ updates. 
This suffices to handle fully-dynamic updates and brings $O(\log (T))$ computation overhead on average.

The reduction is formally depicted in Algorithm \ref{algo:red-amortize}. 
The current elements are distributed to $m+1 = \lceil\log_2 T\rceil + 1$ buckets $B_0, \ldots, B_{m}$.
At the $t$-th update ($t\in [T]$), the algorithm first handles the update (insertion or deletion) on $B_0$ and it is guaranteed by our algorithm that a deletion occurs on $B_0$.
The algorithm then rewinds the computation to $B = B_m \ldots B_{k(t) + 2}$ (Line \ref{line:rewind1}) and then {\sc Insert} elements of $B_0 \cup \cdots \cup B_{k(t) +1}$ in the reverse order (Line \ref{line:insert1}).
Algorithm \ref{algo:red-amortize} re-arranges buckets as follows: It keeps buckets $B_{m}, \ldots , B_{k(t)+2}$ unchanged, it puts the youngest $[2^{i} :2^{i+1} - 1]$ elements in $B$ to bucket $B_i$ ($i \in [0: k(t)]$) and the remaining elements to $B_{k(t) + 1}$ (see Line \ref{line:sort} -- \ref{line:sort2}).

\begin{algorithm}[!htbp]
\caption{Reduction -- Amortized to worst case}
\label{algo:red-amortize}
\begin{algorithmic}[1]
\State Initialize $B_i \leftarrow \emptyset $ ($i \in [0: m]$) \Comment{$m = \lceil \log_2 T\rceil$}
\For{$t = 1,2, \ldots, T$}
\State Rewind the computation to $B_{m}\ldots B_{k(t)+2}$  \label{line:rewind1}
\State Add/remove the element in $B_0$ \Comment{Deletion is guaranteed to occur in $B_0$} \label{line:insert-deletion}
\State $B \leftarrow B_0 \cup \cdots \cup B_{k(t)+1}$
\State $B_i \leftarrow B[2^{i}:2^{i+1} - 1]$ ($i \in [0:k(t)]$) \label{line:sort} \Comment{$B_i$ contains the youngest $[2^i: 2^{i+1}-1]$ elements}
\State $B_{k(t)+1} \leftarrow B \setminus (B_0 \cup \cdots \cup B_{k(t)})$ \label{line:sort2}
\State {\sc Insert} $B_{k(t)+1},\ldots, B_0$ 
\Comment{Elements are inserted in reverse order of deletion}
\label{line:insert1}
\EndFor
\end{algorithmic}
\end{algorithm}

For any $t \in  [T]$, let $E_t$ be all existing elements at the end of $t$-th update, and let $E_{t, r}$ be the youngest $r$ elements in $E_t$ (when there are less than $r$ elements in $E_t$, we take $E_{t, r} = E_t$).
The following lemma formalizes the main invariant for Algorithm~\ref{algo:red-amortize}.
\begin{lemma}
\label{lem:size1}
At the end of $t$-th update ($t\in [T]$), one has
\begin{itemize}
\item $|B_0 \cup \cdots \cup B_{k(t)+ 1}| \leq 2^{k(t)+2} + 2^{k(t)}$;
\item $E_{t, 2^{i+1} - 1} \subseteq B_0 \cup \cdots \cup B_{i}$ for any $i \in [0: k(t)]$.
\end{itemize}
\end{lemma}
\begin{proof}
For the first claim, at the end of $(t - 2^{k(t)})$-th update, Line \ref{line:sort} of Algorithm \ref{algo:red-amortize} guarantees that $|B_{i}| \leq 2^{i}$ holds for any $i \in [0:k(t)+1] $ as $(t - 2^{k(t)})$ is a multiple of $2^{k(t)+1}$. Since there are at most $2^{k(t)}$ insertions between the $(t - 2^{k(t)})$-th and $t$-th update, we have 
\[
|B_0 \cup \cdots \cup B_{k(t) + 1}| \leq 2^{k(t)} + \sum_{i=0}^{k(t)+1}2^{i} \leq 2^{k(t)+2} + 2^{k(t)}.
\]

We prove the second claim by induction. The case of $t =1$ holds trivially and suppose the induction holds up to $t -1$. 
Then we know that 
\[
E_{t - 2^{k(t)}, 2^{k(t)+2} - 1} \subseteq B_0 \cup \cdots \cup B_{k(t) + 1}
\] 
at the end of $(t - 2^{k(t)})$-th update as $(t - 2^{k(t)})$ is a multiple of $2^{k(t)+1}$. 
As a consequence, one has 
\[
E_{t - 2^{k(t)}}\setminus E_{t - 2^{k(t)}, 2^{k(t)+2} - 1} \cap E_{t, 2^{k(t)+ 1} - 1}  = \emptyset
\]
as there are at most $2^{k(t)}$ deletions between the $(t - 2^{k(t)})$-th and $t$-th update, and $2^{k(t)+2} - 1 - 2^{k(t)} \geq 2^{k(t)+1} - 1$. 
At the same time, new elements arrive between the $(t - 2^{k(t)})$-th and $t$-th update and they are all contained in $B_0 \cup \cdots \cup B_{k(t)+1}$, hence, we conclude that $E_{t, 2^{k(t)+1} - 1} \subseteq B_0 \cup \cdots \cup B_{k(t)+1}$. By Line \ref{line:sort} of Algorithm \ref{algo:red-amortize}, we have $E_{t, 2^{i+1} - 1} \subseteq B_0 \cup \cdots \cup B_{i}$ for any $i \in [0: k(t)]$. We conclude the proof here.
\end{proof}

The correctness of Algorithm \ref{algo:red-amortize} follows immediately from the correctness of the incremental algorithm because the set $B_0$ always contains the youngest element. It remains to bound the amortized running time.

\begin{lemma}
\label{lem:amortize}
The amortized update time of Algorithm \ref{algo:red-amortize} is at most $O(\Gamma_{u} \cdot \log (T))$. 
\end{lemma}
\begin{proof}
Line \ref{line:insert-deletion} takes constant time since $B_0$ has constant size.
By Lemma \ref{lem:size1}, $|B_0 \cup \cdots \cup B_{k(t)+1}| \leq 2^{k(t)+2} + 2^{k(t)}$, the rewinding step (Line \ref{line:rewind1}) takes at most $O(2^{k(t)}\Gamma_u)$ time and we make $O(2^{k(t)})$ calls to {\sc Insert} at Line \ref{line:insert1}. The allocation step (Line \ref{line:sort}) takes no more than $O(2^{t(k)})$ time, since the bucket $B_{i}$ has already been sorted ($i \in [0:k(t)+1]$) and it remains to merge them. The total update time equals
\[
\sum_{t=1}^{T} O(2^{k(t)} \Gamma_{u})  = \sum_{k=0}^{m}O(2^{k} \Gamma_{u}) \cdot T/2^{k} = O(\Gamma_{u}\cdot T \log (T)).
\]
We conclude the proof here.
\end{proof}

\begin{Remark}[Implementation of rewinding]
We work in the RAM model and perform reversible computation. One simple way of implementing reversible computation (e.g. \cite{bennett1973logical}) is to write down the change to memory cell in every step. 
The forward computation time only slows down by a constant factor and the backward (rewind) computation time equals the forward computation time. 
In practice for specific problems there may be faster ways to implement rewinding.
\end{Remark}

\section{Reduction -- Worst case to worst case}
\label{sec:worst-case}

We next dedicate to prove Theorem \ref{thm:main}, which translates the worst case guarantee from the incremental model to the deletions-look-ahead model. 
The major difference with the amortized reduction is that one cannot re-order/re-insert a large block of elements at once. 
A natural idea is to prepare the re-order/re-insertion in advance and split the cost.
This brings new challenges as (unknown) future insertions/deletions interleave with the preparation step and one does not know the exact set of elements beforehand.
To resolve it, we maintain multiple threads, and each thread further divides the preparation step into epochs of geometrically decreasing size.

The high-level idea is presented in Algorithm \ref{algo:red} with implementation details deferred to the proof of Lemma \ref{lem:worst}. 
Algorithm \ref{algo:red} maintains $m+1$ threads and $m+1$ buckets $B_0, \ldots, B_{m}$.
During the execution of the algorithm, all existing elements are distributed over $B_0, \ldots, B_m$, and ideally, the $i$-th bucket $B_{i}$ should be of size $O(2^{i})$ and $B_0 \cup \cdots \cup B_{i}$ should contain the youngest $\Omega(2^{i})$ elements.
This guarantees the insertion/deletion of an element can be resolved in $O(\Gamma_u)$ time since one only needs to re-insert elements in  $B_0$.

The crucial part is to maintain the ordered buckets $B_{0}, \ldots, B_{m}$, for which Algorithm \ref{algo:red} maintains $m$ threads; the $i$-th thread ($i \in [m]$) prepares the re-order/re-insertion ahead of $2^{i}$ updates.
Precisely, the $i$-th thread re-starts every $2^{i+1}$ updates and operates over the upcoming $2^{i}$ updates (Line \ref{line:restart1} -- \ref{line:restart2}).
It first rewinds the computation status to $B_m \ldots B_{i}$ (Line \ref{line:rewind}), which is prepared by the $(i + k(\tau))$-th thread and then re-inserts elements of $B_0 \cup \cdots \cup B_{i-1}$ (Line \ref{line:epoch} -- \ref{line:leftover}).
Concretely, the re-insertion procedure is further divided into $i$ epochs, where epoch $j$ ($j \in [i-1]$) lasts for $2^{j}$ updates and epoch $0$ lasts for $2$ updates.
Let $t(i, \tau, j) := 2^{i}\tau + \sum_{r= j+1}^{i-1}2^{r}$ denote the end of epoch $j+1$ in the $(\tau/2 + 1)$-th outer-for-loop-iteration of the $i$-th thread.
During epoch $j$, the $i$-th thread leaves alone the youngest $2^{j+2}$ elements at the beginning of epoch $j$ to $B^{(i)}$ and {\sc Insert} the remaining elements $B_{j}^{(i)}$ over the upcoming $2^{j}$ updates (i.e. $[t(i, \tau, j)+1: t(i, \tau, j) + 2^j]$-th update, see Line \ref{line:insert}).
Meanwhile, the set $B^{(i)}$ is updated and elements are added and removed (Line \ref{line:update}).
Finally, at the end of $t$-th update ($t \in [T]$ and $k(t) \geq 1$), Algorithm \ref{algo:red} resets $B_{j}$ to $B_{j}^{(k(t))}$ for every $j \in [0:k(t)-1]$ (Line \ref{line:init2}) and this step takes $O(\log T)$ time as we only change the pointer of $B_{j}$.


\begin{algorithm}[!htbp]
\caption{\underline{Reduction -- Worst case to worst case}\\
$\triangleright$\ \text{Variables with superscript $^{(i)}$ internal to {\sc Thread}($i$)}\\
$\triangleright$\ \text{{\sc Thread}($i$) only uses information from bigger threads (aka {\sc Thread}($j$) for $j>i$})\\
$\triangleright$\ \text{In particular, {\sc Thread}($i$)'s {\sc Insert} and rewind do not affect the state seen by bigger threads}\\
$\triangleright$\ \text{The output of the algorithm is maintained by {\sc Thread}($0$)}
}
\label{algo:red}
\begin{algorithmic}[1]
\State Initialize $B_i \leftarrow \emptyset$ and run $\textsc{Thread}(i)$ ($i \in [0:m]$) \Comment{$m = \lceil \log_2 T \rceil$}\\
\Procedure{Thread}{$0$}
\For{$t = 1,2, \ldots, T$}
\State Add/remove the element in $B_0$  \Comment{$t$-th update (deletions guaranteed to be from $B_0$)}\label{line:delete}
\State Rewind the computation to $B_{m}\ldots B_1$ \Comment{State prepared by {\sc Thread}($k(t-1)$)}
\State \textsc{Insert} $B_0$  \label{line:back}
\EndFor
\EndProcedure
\\
\Procedure{Thread}{$i$} \Comment{$i \in [1:m]$}
\For{$\tau = 0,2, 4, \ldots, \lfloor T/2^i \rfloor$ }  \label{line:restart1} \Comment{Restart every $2^{i+1}$ updates}
\State $B^{(i)} \leftarrow B_{0} \cup \cdots \cup B_{i-1}$  \label{line:init1}
\State Rewind the computation to $B_{m}\ldots B_{i}$  \label{line:rewind} \Comment{State prepared by {\sc Thread}($i+k(\tau)$)}\\
\For{$j = i-1,i-2,\ldots, 1$} \label{line:epoch} \Comment{$j$-th epoch, amortized over $2^{j}$ updates}
\State $B_{j}^{(i)} \leftarrow B^{(i)}[2^{j+2}+1:]$ \label{line:define-Bj} \Comment{Oldest $\Theta(2^{j})$ elements in $B^{(i)}$}
\State $\textsc{Insert}$ $B_{j}^{(i)}$
\label{line:insert}
\State $B^{(i)}\leftarrow B^{(i)}\backslash B_{j}^{(i)}$  \label{line:remove-Bj}
\State Add/remove elements in $B^{(i)}$ \Comment{Updates $[t(i, \tau, j) + 1 : t(i, \tau, j) + 2^{j}]$}
\label{line:update}
\EndFor
\State Add/remove elements of $B^{(i)}$ in the remaining 2 updates \Comment{Epoch $0$} \label{line:leftover-init}
\State $B_0^{(i)} \leftarrow B^{(i)}$ 
\State $\textsc{Insert}$ $B_0^{(i)}$  \label{line:leftover}\\

\State $B_j \leftarrow B_{j}^{(i)}$ ($\forall j \in [0: i - 1]$) \label{line:init2} 
\State Do nothing for $2^{i}$ updates
\EndFor \label{line:restart2}
\EndProcedure
\end{algorithmic}
\end{algorithm}


Recall $E_{t, r}$ denotes the youngest $r$ elements and $E_t$ denotes all elements at the end of $t$-th update. 
We use $B_{t, j}$ and $B_{t, j}^{(i)}$ to denote the status of $B_{j}$ and $B_{j}^{(i)}$ at the end of $t$-th update. 
We first formalize the main invariant for Algorithm~\ref{algo:red}. 

\begin{lemma}
\label{lem:bucket_size}
For any thread $i \in [m]$, outer-for-loop-iteration $\tau \in \{0,2,\ldots, T/2^i\}$, at the end of the $(2^{i}\tau)$-th update, we have
\begin{itemize}
    \item $E_{2^{i}\tau, 2^{j+2}} \subseteq B_{2^{i}\tau, 0} \cup \cdots \cup B_{2^{i}\tau, j}$,
    \item  $|B_{2^{i}\tau, j}| \leq 3 \cdot 2^{j+1}$
\end{itemize}
holds for any $j \in [i-1]$. For $j = 0$, we have $|B_{2^{i}\tau, 0}| \leq 12$ and $E_{2^{i}\tau, 4} \subseteq B_{2^{i}\tau, 0}$.
\end{lemma}

\begin{proof}
We prove the first bullet by an induction on $i$ (in the reverse order). The base case of $i = m$ holds trivially as all buckets are empty at the beginning.

Suppose the induction holds up to the $(i+1)$-th thread. For any outer-for-loop-iteration $\tau \in \{2, \ldots, T/2^i\}$, at the end of $2^{i}\tau$-th update, the buckets $B_{0}, \ldots, B_{i-1}$ are reset by the $(i+k(\tau))$-th thread, hence, it suffices to prove $E_{2^{i}\tau, 2^{j+2}} \subseteq B_{2^{i}\tau,0}^{(i+k(\tau))}\cup \cdots \cup B_{2^{i}\tau, j}^{(i+k(\tau))}$ ($\forall j \in [i-1]$).
By the inductive hypothesis of the $(i+k(\tau))$-th thread, we know that 
\[
E_{2^{i}\tau -2^{i + k(\tau)}, 2^{i+k(\tau) + 1}} \subseteq B_{2^{i}\tau -2^{i + k(\tau)}}^{(i + k(\tau))} = B_{2^{i}\tau -2^{i + k(\tau)}, 0} \cup \cdots \cup B_{2^{i}\tau -2^{i + k(\tau)}, i + k(\tau) - 1},
\]
that is, the youngest $2^{i+k(\tau) + 1}$ are contained in $B^{(i + k(\tau))}$ initially.
We prove the desired claim by contradiction and assume for some $j \in [0: i - 1]$, there exists an element $e$ such that  $e \in E_{2^{i}\tau, 2^{j+2}}$ but $e\notin B_{2^{i}\tau, 0}^{(i+k(\tau))} \cup \cdots \cup B_{2^{i}\tau, j}^{(i + k(\tau))}$. 
This can only happen if (1) the element $e$ is inserted before epoch $j + 1$; and (2) it is removed from $B^{(i+k(\tau))}$ at some epoch $\gamma \geq j + 1$. 
The reason for (1) is that elements inserted on/after epoch $j+1$ would ultimately be included in $B_{2^{i}\tau, 0}^{(i+k(\tau))} \cup \cdots \cup B_{2^{i}\tau, j}^{(i + k(\tau))}$; the reason for (2) is similar.


Since the element $e$ is removed from $B^{(i+k(\tau))}$ at epoch $\gamma$, we have that 
\[
e \notin E_{t(i + k(\tau), \tau/2^{k(\tau)} - 1, \gamma)} \setminus E_{t(i + k(\tau), \tau/2^{k(\tau)} -1, \gamma), 2^{\gamma + 2}}.
\]
There are at most $2 + \sum_{r=1}^{\gamma}2^{r} = 2^{\gamma + 1}$ deletions since epoch $\gamma$, the rank of $e$ can be improved to at most $2^{\gamma + 2} - 2^{\gamma + 1} = 2^{\gamma+1}$, hence we have $e \notin E_{2^{i}\tau} \setminus E_{2^i \tau, 2^{\gamma + 1}}$ and therefore $e \notin E_{2^{i}\tau} \setminus E_{2^i \tau, 2^{j + 2}}$ ($\gamma \geq j + 1$),  this contradicts with the assumption.


For the second bullet, for any $i \in [m]$ and $\tau \in \{2, \ldots, T/2^i\}$, consider the $(i+k(\tau))$-th thread. 
After executing Line~\ref{line:remove-Bj} in the $(j+1)$-th epoch of Algorithm \ref{algo:red}, we have that 
\[
|B^{(i+k(\tau))}_{t(i + k(\tau), \tau/2^{k(\tau)} - 2^{j+1}, j)} | \leq 2^{j + 3} 
\]
For the rest of the epoch, there can be at most $2^{j+1}$ insertions, hence:
\[
|B^{(i+k(\tau))}_{t(i + k(\tau), \tau/2^{k(\tau)}, j)} | \leq 2^{j + 3} + 2^{j+1}
\]
Finally, after executing Line~\ref{line:define-Bj} in the $j$-th epoch, we have:
\[
|B^{(i+k(\tau))}_{t(i + k(\tau), \tau/2^{k(\tau)}, j), j}| \leq 2^{j + 3} + 2^{j+ 1} - 2^{j+2} = 3 \cdot 2^{j+1}.
\]

We have proved the first claim for $j \in [i-1]$, the case of $j = 0$ follows similarly.
\end{proof}

We next bound the worst case update time.
\begin{lemma}
\label{lem:worst}
The update time per operation is at most $O(\Gamma_{u}\cdot \log^2 (T))$.
\end{lemma}
\begin{proof}
By Lemma \ref{lem:bucket_size}, the size of $B_0$ is $O(1)$ and it contains the youngest $2$ elements, hence the rewinding and {\sc Insert} step (Line \ref{line:back}) can be performed in $O(\Gamma_u)$ time per update.
The major overhead comes from maintaining $m$ threads, and we bound the runtime of each thread separately.

For any thread $i$ and outer-for-loop-iteration $\tau \in \{0,2,\ldots, T/2^{i}\}$, due to Line \ref{line:define-Bj} of Algorithm \ref{algo:red}, we have  
\[
|B_{t(i, \tau, j-1), j}^{(i)}| \leq |B_{t(i, \tau, j)}^{(i)}|  \leq 2^{j+3} + 2^{j+1} \quad \forall j \in [i - 2]
\]  
and by Lemma \ref{lem:bucket_size},
\[
|B_{t(i, \tau, i-2), i-1}^{(i)}| \leq |B_{2^{\tau}i, 0}\cup \cdots \cup B_{2^{\tau}i, i-1}| \leq  \sum_{r=0}^{i-1}3\cdot 2^{r+1} \leq 3\cdot 2^{i+1}.
\]

We analyse the update time step by step.

We first come to the rewinding step (Line \ref{line:rewind}). Unlike the amortized case, we cannot simply rewind by the reversible computation since we maintain multiple threads that need to access the state of the incremental algorithm with different sets of elements, in parallel. Instead, when we call {\sc Insert} of each block $B_{j}^{(i)}$, we maintain a dictionary that records the location/value of changed memory cell. The construction of dictionary only incurs constant overhead. By doing this, during the execution of Algorithm \ref{algo:red}, one can access any memory cell by looking up to at most $O(\log (T))$ dictionaries (note the lookup path is known to Algorithm \ref{algo:red}) and find the last time it has been changed. Naively, looking up the memory updates in each dictionary takes $O(\log (\Gamma_{u}T))$ time. This brings an $O((\log (T)\log (\Gamma_{u}T)))$ total overhead for every operation of {\sc Insert}. Except for this, Line \ref{line:rewind} essentially comes for free.

A more careful implementation leads to only $O(\log (T))$ overhead. 
We maintain an additional data structure, which links each memory cell of the incremental algorithm to $m+1$ lists, where the $i$-th list records the changes made by the $i$-th thread in chronological order. 
The maintenance of the data structure slows down the forward computation of {\sc Insert} by a constant factor. At the same time, in order to search the content of a memory cell, we only need to search through the lists (note again the look-up path is known), which takes $O(1)$ time per list and $O(\log (T))$ in total. Hence, it brings an $O((\log (T))$ total overhead for every operation of {\sc Insert}.

Algorithm \ref{algo:red} updates $B^{(i)}$ and $B_{j}^{(i)}$ at the beginning of epoch $j$ (Lines~\ref{line:define-Bj} and~\ref{line:remove-Bj}). We do not rewrite, but instead, we copy $B^{(i)}$ and $B_{j}^{(i)}$ to new memory cells. Since both sets are of size $O(2^{j})$, the copy operation can be done in the first $\frac{1}{4} \cdot 2^{j} $ updates during epoch $j$ and has $O(\log (T))$ cost per update using Binomial heap.

Algorithm \ref{algo:red} calls at most $O(2^{j})$ times \textsc{Insert} during epoch $j$ (Line \ref{line:insert}). Since elements in $B_{j}^{(i)}$ are known at the beginning so these operations can be averaged over the following $3/4 \cdot 2^{j}$ updates of epoch $j$ and take $O(\Gamma_u \log (T))$ time per update.

The set $B^{(i)}$ receives new elements as well as removes old elements in epoch $j$ (Line \ref{line:update}). We buffer the changes in the first $\frac{1}{4} \cdot 2^{j}$ updates (as the ``new'' set $B^{(i)}$ is not yet ready) and add/remove elements during the following $3/4 \cdot 2^{j}$ updates. The size of $B^{(i)}$ is $O(2^j)$ during epoch $j$, so the update cost is $O(\log (T))$ per update.

Finally, we note (1) Lines \ref{line:leftover-init} -- \ref{line:leftover} takes\ only $O(\Gamma_u \log(T))$ time in total; (2) Line \ref{line:init1} can be done similarly to Line \ref{line:define-Bj}; (3) Line \ref{line:init2} resets $B_j$ ($j \in [0:i -1]$) by changing the pointer, so it incurs only $O(\log T)$ cost.

Overall, Algorithm \ref{algo:red} has worst case update time $O(\Gamma_{u} \log (T))$ per thread and $O(\Gamma_u \cdot  \log^2 (T))$ in total.
\end{proof}

\begin{proof}[Proof of Theorem \ref{thm:main}]
The worst case guarantee has already been established in Lemma \ref{lem:worst}, it remains to prove the correctness of Algorithm \ref{algo:red}. 

By Lemma \ref{lem:bucket_size}, the youngest $2$ elements are always contained in $B_0$, hence insertions/deletions are operated correctly, i.e., the removal step (Line \ref{line:delete}) indeed removes element in $B_{0}$.
It remains to prove each thread operates normally, i.e., for any thread $i$ and outer-for-loop-iteration $\tau$, the removal operation would only remove elements in $B^{(i)}$ during epoch $j$ ($j \in [i-1]$).
It suffices to prove that $E_{t(i, \tau, j), 2^{j+2}} \in B^{(i)}_{t(i, \tau, j)}$.
We prove by induction. This is true in epoch $i-1$ by Lemma \ref{lem:bucket_size}. Suppose it holds to epoch $j+1$, i.e., $E_{t(i, \tau, j+1), 2^{j+3}} \in B^{(i)}_{t(i, \tau, j+1)}$, since there are at most $2^{j+1}$ deletions in epoch $j+1$, we have that $E_{t(i, \tau, j+1), 2^{j+2}} \in B^{(i)}_{t(i, \tau, j)}$. We complete the proof here.
\end{proof}

\section{Application}
\label{sec:application}
We provide a few applications of our reduction. 

\subsection{Submodular maximization} 

\paragraph{Dynamic submodular maximization} In a submodular maximization problem, there is a ground set $N = [n]$ and a set function $f: N\rightarrow \R^{+}$. The function is said to be monotone if $f(A) \geq f(B)$ for any $B \subseteq A \subseteq N$ and it is said to be submodular if $f(A \cup \{u\}) - f(A) \leq f(B \cup \{u\}) - f(B)$ for any $B \subseteq A \subseteq N$ and element $u$. The task of submodular maximization under a cardinality constraint refers to $\max_{S \subseteq [n], |S| = k}f(S)$ for some parameter $1 \leq k \leq n$, and the task of submodular maximization under a matroid constraint $\mathcal{M}$ refers to $\max_{S \subseteq \mathcal{M}}f(S)$. Finally, in a dynamic submodular maximization problem, the ground set can be inserted and deleted, and the goal is to maintain a good solution set $S$.

\cite{feldman2022streaming} provides an $0.3178$-approximation algorithm (with a matroid constraint) under streaming setting, and one can adapt it to a dynamic algorithm with worst case update time under incremental setting.
\begin{theorem}[Adapt from \cite{feldman2022streaming}]
For any $n,k > 0$, under the incremental update, there exists a dynamic algorithm that maintains an $0.3178$-approximate solution for monotone submodular maximization under a matroid constraint of rank $k$ and makes $\poly(k, \log n)$ queries per iteration.
\end{theorem}

The sliding window model is of interests to the community \cite{epasto2017submodular}. 
The algorithm of \cite{epasto2017submodular} maintains an $1/2$-approximation solution with polylogarithmic updates time for dynamic submodular maximization under cardinality constraints.
Our reduction gives the first constant approximation algorithm for a matroid constraint.

\begin{theorem}[Dynamic submodular maximization]
For any $n, k > 0$, there exists a dynamic algorithm that achieves $0.3178$-approximation for the problem of submodular maximization under a matroid constraint using $\poly(k, \log n)$ queries per update under the sliding window model.
\end{theorem}

\subsection{Depth first search (DFS) tree}
\paragraph{Dynamic DFS} Given an undirected graph $G = (V, E)$ with $|V| = n, |E| = m$, the task is to maintain a depth first search (DFS) tree under edge insertion/deletion. In the incremental model, \cite{baswana2016dynamic} obtains a dynamic algorithm with $O(n(\log n)^3)$ worst case update time, and it is improved to $O(n)$ by \cite{chen2018improved}. While in the fully dynamic model, the current best known algorithm \cite{baswana2016dynamic} has $\widetilde{O}(\sqrt{mn})$ update time.

\begin{theorem}[\cite{baswana2016dynamic, chen2018improved}]
Given a graph $G = (V, E)$, with $|V| = n$, $|E| = m$. There is a dynamic algorithm that maintains a DFS tree with $O(n)$ update time in the incremental model. 
\end{theorem}

Using our reduction, one can immediately obtain
\begin{theorem}[Dynamic DFS]
Given a graph $G = (V, E)$, with $|V| = n$, $|E| = m$. There is a dynamic algorithm that maintains a DFS tree with $\widetilde{O}(n)$ worst case update time in the offline model.
\end{theorem}

\subsection{All-Pair Shortest-Paths (APSP)}

\paragraph{Dynamic APSP} The APSP problem has been a central topic of graph algorithm. 
In a dynamic APSP problem, there is a undirected weighted graph $G = (V, E)$ ($|V| = n, |E| = m$) that subjects to edge insertion and deletion, The goal of the algorithm is to maintain an estimate $\delta(u, v)$ for every pair of node $u, v \in V$ that approximates the shortest path distance between $u$ and $v$.
In the incremental setting, \cite{chen2020fast} obtains an $O(1)$-approximate algorithm with $n^{o(1)}$ worst case update time.

\begin{theorem} [Theorem 3.1 in \cite{chen2020fast}]
Let $G = (V, E)$ be a undirected weighted graph, there exists an incremental deterministic All-Pair Shortest-Paths algorithm that maintains $O(1)$ approximate shortest path in $n^{o(1)}$ worst case update time.
\end{theorem}

The offline model is of interest and it is already pointed out in \cite{chen2020fast} that their data structure can be adapted to the offline model (but in a problem specific way).
With Theorem \ref{thm:main} in hand, we can recover Theorem 4.8 of \cite{chen2020fast}.
\begin{theorem}[Dynamic APSP]
Let $G = (V, E)$ be a undirected weighted graph, there exists a deterministic All-Pair Shortest-Paths algorithm that maintains $O(1)$ approximate shortest path in $n^{o(1)}$ worst case update time under the offline model.
\end{theorem}

\section{Impossibility of general reduction}
\label{sec:impossible}
We prove both conditions (worst case guarantee and known deletion order) are indeed necessary to obtain a black box reduction. 

\paragraph{Worst case guarantee is necessary}
A black box reduction is generally impossible if one only has amortized guarantee of incremental model. An example is the dynamic submodular maximization problem, where unconditional lower bound is known.

\begin{theorem}
Let $T \geq 1$ be the total number of updates. There exists a problem such that it is possible to have a dynamic algorithm with amortized update time $\Gamma_u$ in the incremental model, but any algorithm in the deletions-look-ahead model takes at $\Omega\left(\frac{T}{\log^4 (T)}\right) \cdot \Gamma_{u}$ amortized update time.
\end{theorem}
\begin{proof}
Let $[n]$ be the ground set and the total number of update be $T = O(n)$.
By Theorem 1.3 of \cite{chen2022complexity}, there exists an algorithm with $O(\log(k/\eps)/\eps^2)$ amortized query complexity and maintains an $(1-1/e-\eps)$-approximate solution for dynamic submodular maximization.
While in the fully dynamic model (with known deletion order), by Theorem 1.2 of \cite{chen2022complexity}, no algorithm could maintain an $0.584$-approximation with $o(n/k^3)$ amortized queries whenever $k = \Omega(\log n)$. 
Taking $k = C \log n$ for some constant $C > 0$ exhibits a $\Omega(T/\log^4 (T))$ separation.
\end{proof}

\paragraph{Known deletion order is necessary}
If the deletion order is not known in advance, there exists a separation between the fully-dynamic and incremental model.

\begin{theorem}
Let $T \geq 1$ be the total number of updates. There exists a problem such that it is possible to have a dynamic algorithm with worst case update time $O(1)$ in the incremental model, but any algorithm in the fully dynamic model has amortized running time $\Omega\left(\frac{T}{\log (T)}\right)$.
\end{theorem}
\begin{proof}
Let $N = [n]$ be the ground set element and $T = 2n$ be the total number of updates.
We first formalize the oracle model. 
For any subset of elements $A \subseteq N$, let $S(A) \in S$ be the transcript on $A$ and $Q(A) \in \{0, 1\}$ be the answer to {\sc Query}. 
There exists an oracle $O: S \times N \rightarrow S \times \{0, 1\}$, it takes input of a transcript $S(A) \in S$ and an element $e \in N$, and returns the next transcript $S(A \cup \{e\}) \in S$ and the answer $Q(A \cup \{e\}) \in \{0, 1\}$, i.e.,
\[
O(S(A), e) = (S(A \cup \{e\}),Q(A \cup \{e\})).
\]
The oracle outputs empty when the input is invalid and we assume the oracle takes unit time.

It is clear that it takes only $1$ oracle call per update in the incremental model. 
For the fully dynamic model, consider the update sequence of first inserting all elements in $N$ and delete them in random order.
We prove $\Omega(n^2/\log n)$ oracle calls are necessary.
For any $t \in [0: n/100\log n]$, let $N_{t}$ be the set of elements remain after deleting $4t \log n$ elements and $N_0 = N$.
It suffices to prove the algorithm needs to make at least $\Omega(n)$ oracle queries between $N_t$ and $N_{t+1}$. 
Suppose the algorithm has the transcript $S(A_1), \ldots S(A_{\ell(t)})$ at the beginning of the $(n + 4t\log n)$-th update, where $A_i \subseteq [N_t]$ and $|A_i| \geq n/2$ ($i \in [\ell(t)]$). 
We know that $\ell(t) \leq n^2$.
After deleting the next $4\log n$ elements (at random), the probability that $A_i \subseteq N_{t+1}$ is at most $1/n^4$. Taking an union bound, with probability at least $1-1/n^2$, none of the set satisfies $A_{i} \subseteq N_{t+1}$.
Then we conclude that to get the transcript of $S(N_{t+1})$, the algorithm needs at least $|N_{t+1}| - n/2 = \Omega(n)$ queries. 
We conclude the proof here.
\end{proof}

\section*{Acknowledgement}
A.R. is grateful to Amir Abboud and Soheil Behnezhad for inspiring conversations.
A.R. and B.P would like to thank Timothy Chan for explaining the interesting connections to related work from computational geometry.

\bibliographystyle{alpha}
\bibliography{ref}

\end{document}